\documentclass{article}
\usepackage{amsfonts,amsmath,amsthm}
\usepackage[utf8]{inputenc}

\let\phi\varphi
\def\H{\mathop{\hbox{\bf H}}}
\newcommand\M{\mathcal M}
\newcommand\N{\mathcal N}
\newcommand\A{\mathcal A}
\newcommand\down{{\downarrow}}
\newcommand\GM{\Gamma_{\!M}}
\newcommand\MMRV{{\sf MMRV}}
\newcommand\rr{\hbox{\sf r}}

\def\H{\mathop{\hbox{\bf H}}}
\newtheorem{theorem}{Theorem}
\newtheorem{lemma}[theorem]{Lemma}
\newtheorem{claim}[theorem]{Claim}
\newtheorem{corollary}[theorem]{Corollary}

\newtheorem*{definition}{Definition}

\newtheorem{conjecture}{Conjecture}
\newtheorem*{problem}{Problem}

\title{\bf Secret sharing and duality}
\author{Laszlo Csirmaz\thanks{Central European University, Budapest}}
\date{\it\small
Dedicated to the memory of Frantisek Mat\'u\v s}

\begin{document}

\maketitle

\begin{abstract}
Secret sharing is an important building block in cryptography.
All explicit secret sharing schemes which are known to have optimal
complexity are
multi-linear, thus are closely related to linear codes. The dual of such a
linear scheme, in the sense of duality of linear codes, gives another scheme
for the dual access structure. These schemes have the same complexity,
namely the largest share size relative to the secret size is the same. It is
a long-standing open problem whether this fact is true in general: the
complexity of any access structure is the same as the complexity of its
dual. We give a partial answer to this question. An almost perfect scheme
allows negligible errors, both in the recovery and in the independence.
There exists an almost perfect ideal scheme on 174 participants whose 
complexity is strictly smaller than that of its dual. 


\end{abstract}

\section{Introduction}

The complexity of a secret sharing scheme is the largest share size relative
to the secret size. An access structure is ideal if it can be realized by a
scheme of complexity 1. The open question that has been posed in many papers is
if there exists an ideal structure whose dual is not ideal. And, more
generally, if the optimal complexity of an access structure is preserved by
duality. This paper gives a partial answer to these questions by using a
different secret sharing model. This model allows negligible errors, both in
secret recovery and in the independence; ``almost'' refers to this
relaxed model.
The construction of a secret sharing scheme whose almost complexity
differs from its dual's one is a \emph{tour de force} connecting several different
pieces of earlier results.
Theorems \ref{thm:main-ideal} and \ref{thm:main-almost-ideal} state the equivalence
of secret sharing conjectures and matroid representation problems
using the standard and the relaxed models, respectively. The final
construction in Section \ref{sec:final} is based on the second theorem.
Settling the duality conjecture in the standard model is an interesting
research work, a possible direction is indicated in the last section.

We assume familiarity with secret sharing schemes, for an overview 
consult~\cite{beimel-survey}. A significant portion of matroid and polymatroid
theory is used. The standard textbook for matroids is \cite{oxley}, for
polymatroids see \cite{lovasz} and works of F.~Mat\'u\v s
\cite{M.back,fmadhe}. Nevertheless, most of the theorems and claims are
proved here -- a notable exception is F.~Mat\'u\v s result from \cite{fmtwocon}.

Following the usual practice, sets and their subsets are denoted by capital
letters, their elements by lower case letters. The union sign $\cup$ is
frequently omitted as well as the curly brackets around singletons. Thus
$asP$ denotes the set $\{a,s\}\cup P$. The set difference operator has lower
priority than the union, thus $aA{-}bB$ is $(\{a\}\cup A) {-} (\{b\}\cup
B)$.

The paper is organized as follows. Section~\ref{sec:preliminary}
introduces polymatroids, secret sharing, complexity measures, duality, and
concludes with conjectures on the complexity of dual structures.
Section~\ref{sec:ideal-matroid} presents two questions on matroid
representability and proves that they are equivalent to the
conjectures. Section~\ref{sec:final} gives a detailed account of Tarik
Kaced's result on almost entropic matroids \cite{tarik}, completing the
tour. Some open problems are listed in Section \ref{sec:conclusion}.
Two proofs are postponed to the Appendix: the first is
on matroid circuits used in Claim~\ref{claim:connected-to2}, the second is
the the MMRV entropy inequality used in the proof of Theorem~\ref{thm:main}.

\section{Preliminaries}\label{sec:preliminary}

\subsection{Polymatroids}\label{subsec:polymatroids}

A \emph{polymatroid} $\M=(f,M)$ is a non-negative, monotone and submodular
function $f$ defined on the collection of non-empty subsets of the finite
set $M$. Here $M$ is the \emph{ground set}, and $f$ is the \emph{rank
function}. If $f$ takes non-negative integer values only, then $\M$ is
\emph{integer}; an integer polymatroid is a \emph{matroid} if the rank of
singletons are either zero or one. Polymatroids can be identified to
vectors in the $(2^{|M|}-1)$-dimensional Euclidean space where the
coordinates are indexed by subsets of $M$. The collection of
polymatroids with ground set $M$ is a full-dimensional pointed polyhedral
cone denoted by $\GM$, see \cite{yeung-course}.

\smallskip

For a discrete random variable $\xi$ its information content is measured by
the Shannon entropy $\H(\xi)$, see \cite{yeung-course}. Let $\xi=\langle \xi_i:
i\in M\rangle$ be a collection of discrete random variables with some joint
distribution. For a subset $A\subseteq M$, the subcollection $\langle \xi_i:
i\in A\rangle$ is denoted by $\xi_A$. The \emph{conditional entropy} of
random variables $\xi_A$ and $\xi_B$ is $\H(\xi_A|\xi_B) = \H(\xi_{A\cup B})
- \H(\xi_B)$ with value between zero an $\H(\xi_A)$. The value is zero if and
only if $\xi_A$ is determined completely by $\xi_B$, and equals $\H(\xi_A)$
if and only if the random variables $\xi_A$ and $\xi_B$ are independent.

As observed by Fujishige \cite{fuji}, the function $A \mapsto \H(\xi_A)$ is
a rank function of a polymatroid which we denote by $\M_\xi$. The
polymatroid $\M$ is \emph{entropic} if it can be got this way. The
collection of entropic polymatroids on the ground set $M$ is $\GM^*\subseteq
\GM$. For $|M|\ge 3$ the set $\GM^*$ is not closed (in the usual Euclidean
topology). Polymatroids in the closure of $\GM^*$ are called \emph{almost
entropic}, or just \emph{aent}. Aent polymatroids form a full-dimensional
convex cone, and every internal point of this cone is entropic \cite{fmtwocon}. For
$|M|\ge 4$ there is a polymatroid in $\GM$ with a positive distance from the
aent cone \cite{yeung-course}; and the aent cone is not polyhedral
\cite{fminf}.

By an abuse of notation, we say that $\M$ is an \emph{entropic matroid}
if $\M$ is a matroid and for some positive real number $\lambda$
the polymatroid $\lambda\M$ is entropic. 

\smallskip

The singleton $e\in M$ in the polymatroid $(f,M)$ is a \emph{loop} if it has
rank zero. In terms of entropic polymatroid being a loop means that the
variable $\xi_e$ is deterministic: takes a single value with probability 1.
If not mentioned otherwise, polymatroids in this paper have no loops.

With an eye on entropic polymatroids, disjoint subsets $A,B$ of the ground
set $M$ are called \emph{independent} if $f(AB)=f(A)+f(B)$. If $A$ and $B$ are
independent, $A'\subseteq A$, 
$B'\subseteq B$, then  $A'$ and $B'$ are independent as well
 -- this follows from the submodularity of the rank function.
The single subset $A$ is \emph{independent} if any two disjoint subsets
of $A$ are independent. In other words, $A$ is independent iff
$$
    f(A) = {\textstyle\sum}\,\{ f(i): \,i\in A\}.
$$
A \emph{base} is a maximal independent subset which contains no loops; a
\emph{circuit} is a minimal dependent subset. In a loopless polymatroid 
every independent set
can be extended to a base, and every dependent set contains a circuit. In
the case when $\M$ is a matroid every base has the same number of elements,
and this number equals the rank of the ground set $M$. Moreover every subset
$A\subseteq M$ contains an independent set of size $f(A)$, and every subsets
$A\subseteq M$ with rank $f(A)<|A|$ contains a circuit, see \cite{oxley}.

\smallskip

The polymatroid $(f,M)$ is \emph{connected} if for every partition
of $M$ into two non-empty sets $A$ and $B$ we have $f(A)+f(B)>f(M)$, that
is, $A$ and $B$ are \emph{not} independent.
Connected polymatroids have no loops. Indeed, if $i\in M$ is a loop then 
$f(M)=f(M{-}i)$, thus the partition $M=\{i\}\cup (M{-}i)$ contradicts the
connectedness.

\smallskip

For an element $i\in M$, the \emph{private info of $i$} is $f(M)-f(M{-}i)$,
as this is the amount of information which only $i$ and nobody else in $M$
has. If $i$ has no private information, then we say that the polymatroid is
\emph{tight at $i$}. \emph{Tightening at $i$} means that $i$ is stripped off
its private info resulting in the function $f\down i$ defined as
$$
    f\down i : A \mapsto \begin{cases}
          f(A) & \mbox{if $i\notin A$} \\[2pt]
          f(A)-\big(f(M)-f(M{-}i)\big) & \mbox{if $i\in A$}.
    \end{cases}
$$
Of course, $(f\down i,M)$ is a polymatroid tight at $i$. If $\M=(f,M)$ is
tight at every $i\in M$ then $\M$ is \emph{tight}. $\M\down$ is the
polymatroid got from $\M$ after tightening at every element of its ground
set (the result is independent of the order the elements are taken).
Clearly, $\M$ is tight if and only if $\M=\M\down$. If $\M$ is almost
entropic then $\M\down$ is almost entropic; this is a result of F.~Mat\'u\v
s \cite[Lemma~3]{entreg}. In particular, the tight part of an entropic
polymatroid is guaranteed to be almost entropic, but it is not necessarily
entropic. A notable exception is the case of matroids: a matroid $\M$ is
entropic if and only if $\M\down$ is entropic. It is so as if $i$ is not
tight, then $1=f(M)-f(M{-}i) \le f(i)\le 1$ thus $i$ is independent from
all subsets of $M{-}i$, thus the random variable representing $i$ can be 
discarded.

\subsection{Secret sharing}\label{subsec:secret-sharing}

In a perfect secret sharing scheme there is a \emph{secret}, and each participant
from the finite set $P$ receives a \emph{share} such that certain subsets of
participants can recover the secret from their joint shares, while other
subsets -- based on the value of their shares -- should have no
information on the secret. Subsets who can recover the secret are
\emph{qualified}, the qualified subsets form the \emph{access structure
$\A\subseteq 2^P$}. Sets not in $\A$ are called \emph{forbidden} or
\emph{unqualified}. An access structure is clearly upward closed. To avoid
exceptional cases, $\A$ is assumed to be non-empty (thus all participants
together can recover the secret), and the empty set not to be in $\A$ 
(there must be a secret at all).

The participant $i\in P$ is \emph{important} if there is an unqualified
subset such that when $i$ joins this subset, it becomes qualified. If $i$ is
not important, then it can join or leave any subset without affecting its
status. Consequently the share of an unimportant participant does not play
any role, unimportant participants can be discarded. The access
structure $\A$ is \emph{connected} if every participant is important. This
terminology comes from the relationship between access structures and
polymatroids realizing them, see Claims~\ref{claim:connected-to1} and
\ref{claim:connected-to2} below. In the rest of the paper, if not mentioned
otherwise, access structures are assumed to be connected.

\smallskip

There are several definitions of what perfect secret sharing schemes are. The
following definition is considered to be the most general one encompassing
all other natural notions \cite{beimel-survey}. $P$ is the set of
participants and $s\notin P$ denotes the secret.
A \emph{distribution scheme} is a collection of discrete random variables
$\xi=\langle \xi_i:i\in sP\rangle$ with some joint distribution. The value of
$\xi_s$ is the \emph{secret}, while the value of $\xi_i$ is the \emph{share}
of participant $i\in P$. The secret must be non-trivial, namely it must take at
least two different values with positive probability.

The distribution scheme $\xi$ \emph{realizes} an access structure if a) the
collection of shares of a qualified subset determine the secret, and b) the
collection of shares of an unqualified subset is independent of the secret.
Let $\M_\xi=(f,sP)$ be the entropic polymatroid associated with $\xi$.
Shares of the subset $A\subseteq P$ determine the secret iff $\H(\xi_s |
\xi_A)=0$, which translates to $f(sA) = f(A)$. The same collection is
independent of the secret if $\H(\xi_s|\xi_A)=\H(\xi_s)$, which translates
to $f(sA)=f(A)+f(s)$. This justifies the following definition.

\begin{definition}[realizing an access structure]\upshape
The polymatroid $\M=(f,sP)$ \emph{realizes the access structure $\A\subseteq
2^P$} if a) $A\in\A$ if and only if $f(sA)=f(A)$, and b) $A\notin\A$ if and
only if $f(sA)=f(A)+f(s)$. Polymatroids realizing an access structure are
called \emph{secret sharing polymatroids}.
\end{definition}

The entropic polymatroid $\M_\xi$ realizes the access structure $\A$ if
and only if $\xi$ is a distribution scheme realizing $\A$. Indeed, if
$\xi$ is a distribution scheme then $f(s)=\H(\xi_s)$ is positive, thus
one cannot have $f(sA)=f(A)$ and $f(sA)=f(A)+f(s)$ at the same time.
Conversely, if $\M_\xi$ realizes $\A$, then $f(s)>0$ (otherwise
both $f(As)=f(A)$ and $f(As)=f(A)+f(s)$ hold simultaneously),
thus the secret is not trivial. Other conditions follow easily.

\smallskip

The proof of the following well-known fact illustrates the ease of
reasoning when using polymatroids rather than using entropies directly.

\begin{claim}\label{claim:minimum-is-1}
Suppose $\M$ realizes $\A$. 
Then $f(i)\ge f(s)$ for every important participant $i\in P$.
\end{claim}
\begin{proof}
As $i\in P$ is important, there is an unqualified subset
$A\subseteq P$ ($A$ can be empty) such that $iA$ is qualified. Then
$f(sA)=f(A)+f(s)$, and $f(siA)=f(iA)$. Using that $f(i)+f(sA)\ge f(siA)$ and
$f(iA)\ge f(A)$ (submodularity and monotonicity) one gets
\begin{align*}
   f(i)\ge f(siA)-f(sA) = f(iA)-\big(f(A)-f(s)\big) \ge f(s),
\end{align*}
which proves the claim.
\end{proof}

All participants together can always determine the secret, thus
$f(sP)=f(P)$. This means that the secret has no private info. The private
info of the participants does not help at all.

\begin{claim}\label{claim:realize-down}
The polymatroid $\M$ realizes $\A$ if and only if $\M\down$ realizes $\A$.
\end{claim}
\begin{proof}
As observed above, the secret is tight, so let $i\in P$ and $\M\down
i=(f^*,sP)$ be the polymatroid after taking away the private info of $i$.
For every $A\subseteq P$, either $i$ is in both $A$ and $sA$, or $i$ is in
none of them, thus
$$
    f^*(sA)-f^*(A) = f(sA)-f(A).
$$
This means that if one of $\M$ or $\M\down i$ realizes $\A$, then the other 
does the same. the claim follows after tightening at each participant.
\end{proof}

Given an access structure it would be tempting to consider tight
polymatroids only among those which realize it. But, as was mentioned at the
end of Section~\ref{subsec:polymatroids}, there is no guarantee that the
tight part of an entropic polymatroid is also entropic.

\begin{corollary}\label{corr:ideal-is-tight}
Suppose $\A$ is connected, and $\M$ realizes $\A$. If $f(i)=f(s)$ 
then $i\in P$ is tight.
\end{corollary}
\begin{proof}
By Claim~\ref{claim:realize-down} $\M\down i=(f^*,sP)$ also realizes $\A$.
As $i\in P$ is important, Claim~\ref{claim:minimum-is-1} gives $f^*(i)\ge f(s)$.
Now $f^*(i)\le f(i)=f(s)$, thus $f^*(i)=f(i)$ showing that $i$ is tight.
\end{proof}

According to Claim~\ref{claim:connected-to1} below, a polymatroid realizing a
connected access structure must be connected. The converse is not true in
general. In the special case when the polymatroid is a matroid the converse
follows from some standard properties of matroid circuits \cite{oxley}.

\begin{claim}\label{claim:connected-to1}
Suppose the polymatroid $\M$ realizes the access structure $\A$. 
If the access structure is connected, then $\M$ is connected.
\end{claim}
\begin{proof}
Assume, by contradiction, that $\M=(h,sP)$ is not connected, which means 
$sA\cup B$, $B\not=\emptyset$ is a partition of the ground set $sP$
and $h(sAB)=h(sA)+h(B)$. In other words, $sA$ and $B$ are independent,
consequently subsets of $sA$ and $B$ are independent as well. Let $b\in B$ and
assume $A'B'$ is not qualified while $A'bB'$ is qualified with $A'\subseteq
A$ and $B'\subseteq B{-}b$. Then $h(sA'B')=h(A'B')+h(s)$, from where
$$
   h(sA')+h(B')=h(sA'B')=h(A'B')+h(s)=h(A')+h(B')+h(s).
$$
On the other hand, $h(sA'bB')=h(A'bB')$, which gives
$$
   h(sA')+h(bB') = h(sA'bB')=h(A'bB')=h(A')+h(bB').
$$
From the first line $h(sA')=h(A')+h(s)$, while from the second
$h(sA')=h(A')$, a contradiction.
\end{proof}
\begin{claim}\label{claim:connected-to2}
Suppose $\M$ is a matroid which realizes the access structure $\A$. If $\M$ is
connected then so is the access structure $\A$.
\end{claim}
\begin{proof}
Using matroid terminology, $A$ is \emph{dependent} if $h(A)<|A|$, and $A$ is a
\emph{circuit} if it is a minimal dependent set. If $C$ is a circuit, then $h(C)=|C|-1$
and for each $i\in C$, $h(C{-}i)=|C|-1$. 
A circuit \emph{connects} two points if it contains both of them.

Let the ground set of the matroid be $sP$ and pick some $a\in P$. To show
that $a$ is important it is enough to find a circuit $C$ connecting $a$ and
$s$. Indeed, let $A=C{-}s$, then $a\in A$ and $h(sA)=h(C)=h(C{-}s)=h(A)$,
thus $A$ is qualified. $C$ is minimal dependent, thus $sA{-}a=C{-}a$ is
independent, and then $h(sA{-}a)=h(s)+h(A{-}a)$ which means $A{-}a$ is not
qualified.

To finish the proof it suffices to quote the following result from matroid
theory \cite[Proposition 4.1.4]{oxley}: \emph{a matroid is connected if and only
if any two points can be connected by a circuit}. For a quick proof see the
Appendix.
\end{proof}

\subsection{Complexity}\label{subsec:complexity}

Distribution schemes realizing an access structure scale up: taking $n$
independent copies of the scheme all entropies are multiplied by $n$ and the
composite scheme still realizes the same access structure. Similarly,
whether a polymatroid realizes an access structure or not is invariant for
multiplying the polymatroid by any positive constant. When defining the
efficiency one has to take into account this scalability. The usual way is
to measure everything in multiples of the secret size. For example, if
$\M=(f,sP)$ is a secret sharing polymatroid, then the \emph{relative share
size} of participant $i\in P$ is $f(i)/f(s)$, and the (worst case)
\emph{complexity of $\M$} is
$$
  \sigma(\M) = \max \bigg\{ \frac{f(i)}{f(s)}\,:\, i\in P \bigg\},
$$
Other complexity measures, not considered here, include average relative
size, and the scaled total randomness.
If $\M$ realizes the connected access structure $\A$, then $\sigma(\M)\ge 1$ 
by Claim~\ref{claim:minimum-is-1}. Access
structures where this lower bound is attained are called \emph{ideal}.

\begin{definition}[ideal and almost ideal structures]\upshape
The access structure $\A$ is \emph{ideal} if it can be realized by an entropic
polymatroid with complexity 1.
The access structure $\A$ is \emph{almost ideal} if it can be realized by an
almost entropic polymatroid with complexity 1.
\end{definition}

In general, the usual definition of the complexity of an access 
structure is the infimum of the complexity of all secret sharing schemes
realizing it:
$$
    \sigma(\A) = \inf\,\big\{ \sigma(\M)\,:\,\M
          \mbox{ is entropic and realizes } \A \big\}.
$$
Interestingly, there is a non-ideal (according to our definition) access 
structure with complexity 1, see \cite[Section 6]{beimel-livne}, thus the
infimum here is not necessarily taken.
The cone of almost entropic polymatroids is closed, see \cite{fmtwocon}
or \cite{yeung-course}, thus an access
structure with complexity 1 is almost ideal. It is an interesting open
question whether the converse is true. When approximating an aent
polymatroid by an entropic one, the only guarantee is that the rank
functions differ by a small (negligible) amount. This means that qualified
subsets can recover the secret with ``overwhelming probability'' only (as
$\H(s|A)$ is not necessarily zero, only negligible),
and unqualified subsets might get information on the secret (as $\H(s|A)$
can be strictly smaller than $\H(s)$), but this information is negligible.
This relaxation is investigated under the name \emph{probabilistic secret
sharing} see, e.g., \cite{prob-secret} and \cite{tarik-2}.
The question is can we patch these imperfections by adding a small amount 
of entropy to the secret? For secret recovery the
answer is \emph{yes}, see \cite{tarik-2}; for independence the author tends
to believe that the answer is \emph{no}.

\smallskip

Next to $\sigma(\A)$ other complexity measures can be defined by considering
other polymatroid classes. Realizing $\A$ by an
entropic polymatroid is the same as realizing it by a distribution scheme.
Realizing by an almost entropic polymatroids instead means that one relaxes
the strict
requirements of recoverability and independence ``up to a negligible
amount''. Linearly representable polymatroids are important from both
practical and theoretical point of view. Such polymatroids arise from linear
error correcting codes \cite{error-correcting-codes}, they are studied
extensively and typically provide concise, efficient and low complexity schemes.
We consider the following polymatroid classes, listed in decreasing order:
\begin{itemize}\setlength\itemsep{0pt plus 2pt}
\item[a)] all polymatroids,
\item[b)] almost entropic polymatroids,
\item[c)] entropic polymatroids,
\item[d)] (conic hull of) linearly representable polymatroids.
\end{itemize}
Every access structure can be realized by a linearly representable
polymatroid, thus every class gives a complexity notion on access structures. 
For classes
a), c) and d) they are denoted by $\kappa$, $\sigma$, and $\lambda$ 
\cite{padro-notes}. For class b) we use $\bar\sigma$ to indicate that we are
considering the closure of entropic polymatroids. The earlier definition of
$\sigma(\A)$ is the same as given here.
\begin{align*}
  \kappa(\A) &= \inf\{\, \sigma(\M): \M \mbox{ realizes } \A \},\\
  \bar\sigma(\A) &= \inf\{\, \sigma(\M): \M \mbox{ is aent and realizes } \A \},\\
  \sigma(\A) &= \inf\{\, \sigma(\M): \M \mbox{ is entropic and realizes } \A
\}, \\
  \lambda(\A) &= \inf\{\, \sigma(\M): \M \mbox{ is linear and realizes } \A \}.
\end{align*}
For the same access structure these values increase (as less and less polymatroids are
considered). Each pair of these measures is known to be separated except for $\sigma$ and
$\bar\sigma$, see \cite{beimel-survey,padro-notes}.

\subsection{Duals}

Let $P$ be the set of participants, and $\A\subseteq 2^P$ be an access
structure. The qualified subsets in the \emph{dual access structure
$\A^\bot$} are the complements of unqualified subsets of $\A$:
$$
    \A^\bot = \{ A\subseteq P\,:\, P{-}A\notin \A \}.
$$
Clearly, the dual of  $\A^\bot$ is $\A$; $\emptyset\notin\A^\bot$ and
$P\in\A^\bot$ as these are true for $\A$. 
\begin{claim}\label{claim:access-dual-connected}
$\A$ is connected if and only if $\A^\bot$ is connected.
\end{claim}
\begin{proof}
Suppose $\A$ is connected, we show that $\A^\bot$ is connected. The other
direction follows from $(\A^\bot)^\bot=\A$. 
Let $a\in P$, and $A\subset P$ unqualified 
in $\A$ such that $aA\in\A$. Such an $A$ exists as $\A$ is connected.
Then $a\in P{-}A$, $P{-}A$ is qualified in $\A^\bot$ and
$(P{-}A){-}a$ is not qualified in $\A^\bot$, as required.
\end{proof}

Let $\M=(f,M)$ be a polymatroid.
Define the discrete measure $\mu$ on subsets of $M$ by
$\mu(i)=f(i)$. As the measure is additive, for every
subset $A\subseteq M$ we have
$$
    \mu(A) = {\textstyle\sum}\,\{ f(i): i\in A \}.
$$
The \emph{dual of the polymatroid $\M$} is $\M^\bot=(f^\bot,M)$
where the function $f^\bot$ is defined for subsets of $M$ as
$$
   f^\bot: A \mapsto f(M{-}A)+\mu(A)-f(M).
$$
By submodularity, $f^\bot$ is non-negative; submodularity holds by an easy
inspection, thus $\M^\bot$ is a polymatroid. If $\M$ is integer-valued then so
is $\M^\bot$; moreover if $\M$ is a matroid (the rank of a singleton is zero
or one), then so is the dual.

\begin{claim}\label{claim:poly-dual-connected}
{\upshape a)} $\M$ is connected if and only if $\M^\bot$ is connected.
{\upshape b)} $\M$ realizes the access structure $\A$ if and only if $\M^\bot$
realizes $\A^\bot$.
\end{claim}
\begin{proof}
a) Suppose $A\cup B$ is a partition of $M$, then $\mu(A)+\mu(B)=\mu(M)$. By 
the definition of $f^\bot$ we have
$$
   f^\bot(A)+f^\bot(B)-f^\bot(M)=f(B)+f(A)-f(M).
$$
If one of them is positive, then the other is positive, as required.

b) Let $M=sP$, $A\subseteq P$, then $\mu(sA)-\mu(A)=\mu(s)=f(s)$, thus
$$
   \big(f^\bot(sA)-f^\bot(A)\big) + \big(f(sP{-}A)-f(P{-}A)\big) =f(s).
$$
If $\M$ realizes $\A$, then $f(sP{-}A)-f(P{-}A)$ is either zero or $f(s)$
depending on whether $P{-}A\in\A$ or not. Consequently $f^\bot(sA)-f^\bot(A)$
is either zero of $f(s)$ depending on whether $P{-}A\notin\A$ or not. Thus
$f^\bot$ realizes $\A^\bot$. The converse is similar.
\end{proof}

The dual polymatroid $\M^\bot$ is always tight as
$$
   f^\bot(M{-}i)=f(i)+\mu(M{-}i)-f(M)=\mu(M)-f(M)=f^\bot(M).
$$
Consequently the dual of $\M^\bot$ is also tight, and if $\M$ was not tight,
the dual of $\M^\bot$ cannot be the same as $\M$. However, if $\M$ is tight, then it equals
$\M^\bot{}^\bot$, in particular $\M^\bot{}^\bot{}^\bot=\M^\bot$ always true.

\begin{claim}\label{claim:tight-dual}
{\upshape a)} Suppose $\M$ is tight. Then $\M^\bot{}^\bot=\M$, moreover $\M$ and $\M^\bot$ 
have the same value on singletons.
{\upshape b)} For every polymatroid $\M$, $\M^\bot{}^\bot=\M\down$.
\end{claim}
\begin{proof}
a) We start with the second claim. By the assumption, $f(M)=f(M{-}i)$,
$$
   f^\bot(i)=f(M{-}i)+\mu(i)-f(M)=\mu(i)=f(i),
$$
as claimed. It means that $\mu^\bot(A)=\mu(A)$, and then
\begin{align*}
   f^\bot(M{-}A)&=f(A)+\mu(M{-}A)-f(M), \\
   f^\bot(M)   &= \mu(M)-f(M),
\end{align*}
thus
\begin{align*}
  f^\bot{}^\bot(A)& =f^\bot(M{-}A)+\mu^\bot(A)-f^\bot(M) ={} \\
       &=f(A)+\mu(M{-}A)+\mu(A)-\mu(M) = f(A),
\end{align*}
proving $\M^\bot{}^\bot=\M$.

b) It is enough to show that the dual of $\M$ and the dual of $\M\down$ are
the same, from here the claim follows by a). In $\M\down i$ the rank of every 
set containing $i\in M$ decreases by the same amount. In the expression
$$
   f(M{-}A) + \mu(A)-f(M)
$$
this amount is added once in the first two terms, and subtracted once in the
last term, thus it cancels.
\end{proof}

\subsection{Factor and principal extension}\label{subsec:factors}

Let $\M=(h,M)$ be a polymatroid. Partitions of the ground set
$M$ can be considered as equivalence classes of an equivalence relation on
$M$. Let $\cong$ be an equivalence relation on $M$, $N=M/{\cong}$ be the 
set of equivalence classes, and
$\phi:M\to N$ be the map which assigns to each element its equivalence class.
The \emph{factor of $\M$ by $\cong$}, denoted as $\M/{\cong}$, is the pair
$(g,N)$ where $g$ assigns the value $g: A\mapsto h(\phi^{-1}(A))$ to subsets
of $N$ (that is, union of complete equivalence classes). It is clear that
$\M/{\cong}$ is a polymatroid.

\smallskip
Let $a\in M$, and $\alpha\ge 0$ be a real number.
The \emph{principal extension $\M_{a,\alpha}$} is a one-point extension of
$\M$ defined on the set $M\cup \{a'\}$ assigning the value
$$
    h : a'A \mapsto \min\,\{ h(A)+\alpha, h(aA) \}
$$
to new subsets.
It is a routine to check that the principal extension is a polymatroid
\cite{lovasz}. Principal extension of an almost entropic polymatroid is
almost entropic. This is an immediate consequence of (and actually, is
equivalent to) a result of F.~Mat\'u\v s \cite[Theorem~2]{fmtwocon}, see also
\cite[Lemma 3]{entreg}. We state this result without proof.

\begin{theorem}[F.~Mat\'u\v s]\label{thm:aent-down}
If the polymatroid $\M$ is almost entropic, then so is the principal
extension $\M_{a,\alpha}$.\qed
\end{theorem}

Mat\'u\v s' proof guarantees the extension to be only almost entropic even
if $\M$ is entropic. In fact, there is an entropic polymatroid where some
principal extension is not entropic.

\smallskip

Principal extensions can be used to ``split atoms'' of a
polymatroid, which, in turn, will be used to prove that integer polymatroids
are factors of matroids. Let us see the details. In what follows $\M=(h,M)$
is a polymatroid.

\begin{lemma}\label{lemma:split1}
Let $a\in M$, and $\alpha_1,\alpha_2$ be non-negative numbers whose sum is
$h(a)$. There is a polymatroid $\M'=(h',a_1a_2\cup M{-}a)$ such that
$h'(a_i)=\alpha_i$, and $\M$ is a factor of $\M'$ collapsing $a_1$ and $a_2$
to $a$. Moreover, $\M$ is almost entropic if and only if so is $\M'$.
\end{lemma}
\begin{proof}
Let $\M'$ be the principal extension $\M_{a,\alpha_1}$ adding the new point
$a_1$, so that $M'=a_1\cup M$; then let $\M''$ be the principal extension
$\M'_{a,\alpha_2}$ adding the new point $a_2$. Then for each $A\subseteq
M{-}a$ we have
\begin{align*}
    h''(A) &= h(A), \\
    h''(a_1A) &= \min\,\{ h(A)+\alpha_1, h(aA) \}, \\
    h''(a_2A) &= \min\,\{ h(A)+\alpha_2, h(aA) \}, \\
    h''(a_1a_2A) &= \min\,\{ h(A)+\alpha_1+\alpha_2, h(aA) \}.
\end{align*}
As $h(A)+\alpha_1+\alpha_2=h(A)+h(a) \ge h(aA)$, we have
$h''(a_1a_2A)=h(aA)$. This shows that $\M''$ restricted to the ground set
$a_1a_1\cup M{-}a$ is the required splitting. If $\M$ is aent, then both
$\M'$ and $\M''$ are aent by Theorem~\ref{thm:aent-down}. A restriction and
a factor of an aent polymatroid is trivially aent, proving the last claim.
\end{proof}

\begin{lemma}\label{lemma:splitting-dual}
Let $\M=(f,aN)$ be tight, and suppose $\N=(g,a_1a_2N)$ splits $a$ in $\M$ as
$g(a_i)=\alpha_i$. Then $\N^\bot$ splits $a$ in $\M^\bot$ in the same way.
\end{lemma}
\begin{proof}
Let $A\subseteq N$, then $g(A)=f(A)$ and $g(a_1a_2A)=f(aA)$. Calculating
$g^\bot(A)$ one gets
\begin{align*}
   g^\bot(A)& =g(a_1a_2N{-}A)+\mu(A)-g(a_1a_2N) ={} \\
            &= f(aN{-}A)+\mu(A)-f(aN) = f^\bot(A),
\end{align*}
and similarly $g^\bot(a_1a_2A)=f^\bot(aA)$. Finally,
\begin{align*}
   g^\bot(a_1A) &= g(a_1a_2N{-}a_1A)+\mu(a_1A)-g(a_1a_2N) \\
      &= g(a_2N{-}A) + \mu(a_1)+\mu(A)-f(aN) \\
      &= \min\{ f(N{-}A)+\alpha_2,f(aN{-}A) \}+\alpha_1+\mu(A)-f(aN) \\
      &= \min\{ f(N{-}A)+\mu(aA)-f(aN), f(aN{-}A)+\mu(A)-f(aN)+\alpha_1 \}
       \\
      &= \min \{ f^\bot(aA),f^\bot(A)+\alpha_1 \},
\end{align*}
thus $\N^\bot$ splits $a$ as claimed as $f^\bot(a)=f(a)=\alpha_1+\alpha_2$
using that $\M$ is tight.
\end{proof}

\smallskip

Factors of a matroid are integer polymatroids. Helgason's theorem
\cite{helgason} says that the converse is true: every integer polymatroid is
a factor of some matroid. We need the following strengthening of this result.

\begin{theorem}\label{thm:helgason}
For each integer polymatroid $\M$ there is a matroid $\phi(\M)$ such that
{\upshape a)} $\M$ is a factor of $\phi(\M)$, {\upshape b)} $\M$ is aent if
and only if $\phi(\M)$ is aent, {\upshape c)} if $\M$ is tight, then
$\phi(\M^\bot)$ is the dual of $\phi(\M)$.
\end{theorem}
\begin{proof}
Let $\M$ be an integer polymatroid. The matroid $\phi(\M)$ is generated by a
series of splitting. If all singletons have rank zero or one,
then $\M$ is a matroid, and we are done. Otherwise some $a\in M$ has rank
$h(a)>1$. Using Lemma~\ref{lemma:split1} split $a$ into two with ranks $1$
and $h(a)-1$. All ranks in the split polymatroid $\M'$ remain integer, and by
Lemma~\ref{lemma:split1} $\M'$ is aent if and only if $\M$ is aent. Continue
this way to get the matroid $\phi(\M)$. Clearly $\M$ is a factor of
$\phi(\M)$, and c) holds by Lemma~\ref{lemma:splitting-dual}.
\end{proof}

\subsection{The duality conjecture}

Fix the connected access structure $\A\subset 2^P$ and consider all
polymatroids on the ground set $sP$ which realize $\A$. We are interested
in $\kappa(\A)$, the minimal complexity of these polymatroids. By Claim
\ref{claim:realize-down} the search can be restricted to tight polymatroids.
Suppose the infimum is attained by a tight polymatroid $\M$. (It is
attained as polymatroids form a closed set.)
Claim~\ref{claim:tight-dual} a) implies that $\M$ and $\M^\bot$ have the same
complexity. According to Claim~\ref{claim:poly-dual-connected} b)
$\M^\bot$ realizes $\A^\bot$, thus $\kappa(\A^\bot) \le \kappa(\A)$.
Applying the same reasoning to the dual structure we get
$\kappa(\A^\bot{}^\bot)\le\kappa(\A^\bot)$. As $\A^\bot{}^\bot$ and $\A$ are
the same, we have

\begin{claim}\label{claim:poly-complexity-equal}
For every access structure we have $\kappa(\A)=\kappa(A^\bot)$.
\qed
\end{claim}

Every access structure can be realized by some linearly representable
polymatroid, the complexity measure $\lambda(\A)$ defines the infimum of the
complexity of such representations. It is well-known that the conic hull of
linearly representable polymatroids is a closed subset of the entropic
polymatroids, and it is closed for taking duals. Therefore it is also closed for
tightening by Claim~\ref{claim:tight-dual} b). The corresponding complexity 
measure is $\lambda(\A)$, and the same reasoning as above gives

\begin{claim}\label{claim:linear-complexity-equal}
For every access structure we have $\lambda(\A)=\lambda(\A^\bot)$.
\qed
\end{claim}

Every explicitly defined access structure $\A$ with known exact complexity value
$\sigma(\A)$ satisfies $\lambda(\A)=\sigma(\A)=\kappa(\A)$ -- consequently
the same is true for the dual structure, and then
$\sigma(\A)=\sigma(\A^\bot)$. It is a long-standing open problem whether the
statement similar to Claims~\ref{claim:poly-complexity-equal} and
\ref{claim:linear-complexity-equal} holds for the entropic complexity
$\sigma$.

\begin{conjecture}[complexity of dual structure]\label{conj:dual-or-all}
For every access structure we have $\sigma(\A)=\sigma(\A^\bot)$.
\end{conjecture}

The conjecture is probably not true, but even the particular case when $\A$
is an ideal access structure resisted all efforts. Recall, that $\A$ is
ideal if it can be realized by an ideal entropic polymatroid, or,
equivalently, by an ideal distribution scheme.

\begin{conjecture}[dual of ideal structure]\label{conj:dual-of-ideal}
The dual of an ideal access structure is ideal.
\end{conjecture}

Refuting the second conjecture does not necessarily refutes Conjecture
\ref{conj:dual-or-all} as the dual might be non-ideal while having
complexity 1. In Section~\ref{sec:ideal-matroid} we prove that Conjecture
\ref{conj:dual-of-ideal} is equivalent to a question about matroid
representability. Using results of that section, and a construction by Tarik
Kaced \cite{tarik} the duality question for almost ideal schemes is settled.

\section{Ideal structures and matroids}\label{sec:ideal-matroid}

First we give a self-contained proof of a somewhat extended result of
Blakley and Kabatianski \cite{blakley-kabatianski,fmquantoids},
which, in turn, extends a result of Brickell and Davenport
\cite{brickell-davenport} connecting ideal access structures and matroids.
Using this connection we present a
statement about matroid representability which is equivalent to Conjecture
\ref{conj:dual-of-ideal}.

Fix the connected access structure $A\subset 2^P$ and suppose the polymatroid
$\M=(f,sP)$ realizes it. Assume furthermore that $\M$ has
complexity 1, that is, the rank of all singletons equals $f(s)$.
The following lemmas establish some structural properties of $\M$. In the
lemmas $A$ is a subset of $P$, $a\in P$, and $s$ denotes the secret.

\begin{lemma}\label{lemma:1}
Suppose $A\in\A$ and $A{-}a\notin\A$. Then $f(A)-f(A{-}a)=f(s)$.
\end{lemma}
\begin{proof}
By submodularity of the rank function $f$, we have
\begin{align*}
  f(a) &\ge f(A)-f(A{-}a) = f(sA)-\big(f(sA{-}a)-f(s)\big) = {} \\
       &= \big(f(sA)-f(sA{-}a)\big) +f(s) \ge f(s).
\end{align*}
As $f(a)=f(s)$, the conclusion follows.
\end{proof}
\begin{lemma}\label{lemma:2}
Let $a\in A'\subseteq A$, Suppose $A{-}a$ and $A'$ are qualified and $A'{-}a$
is not. Then $f(A)=f(A{-}a)$.
\end{lemma}
\begin{proof}
For qualified subsets $f(sA')=f(A')$, etc., for the unqualified subset
$f(sA'{-}a) =f(A'{-}a)+f(s)$. Thus
\begin{align*}
  f(A)-f(A{-}a) &= f(sA)-f(sA{-}a) \le f(sA')-f(sA'{-}a) ={}\\
  &= \big( f(A')-f(A'{-}a)\big) - f(s) = 0,
\end{align*}
where the inequality follows from submodularity and the last equality from
Lemma~\ref{lemma:1}. Thus $0 \le f(A)-f(A{-}a)\le
0$, proving the claim.
\end{proof}
\begin{lemma}\label{lemma:3}
Suppose $f(A)-f(A{-}a)=f(s)$, and $a\in A'\subseteq A$. Then
$f(A')-f(A'{-}a)=f(s)$.
\end{lemma}
\begin{proof}
By submodularity, $f(A)-f(A{-}a)\le f(A')-f(A'{-}a)\le f(a)$. As both sides
equal $f(s)$, the claim follows.
\end{proof}

\begin{theorem}[Blakley--Kabatianski]\label{thm:b-k}
Let $\A\subset 2^P$ be a connected access structure and $\M=(f,sP)$ be a
polymatroid realizing $\A$ such that $f(a)=f(s)=1$ for all $a\in P$.
Then $\M$ is a matroid which is uniquely determined by the access structure.
\end{theorem}
\begin{proof}
All singletons have rank 1, thus $\M$ is a matroid if all ranks are
integer. The basic idea is to show that for any subset $A$ of
the ground set $sP$ one can find an element $a$ of $A$ such that
$f(A)-f(A{-}a)$ is either zero or one. The additional claim that $\M$ is
uniquely determined by $\A$ follows from the fact that for the chosen 
element $a\in A$ the value of $f(A)-f(A{-}a)$ depends only on the 
access structure, and not on the particular realization.

If the subset contains the secret $s$, then $f(sA)-f(A)$ is either zero or
$f(s)=1$ depending on whether $A\in\A$ or $A\notin\A$, which settles this
case. So assume $A\subseteq P$. 

When $A$ is qualified, then there are two cases. If $A{-}a$ is not qualified
for some $a\in A$, then Lemma~\ref{lemma:1} gives that this difference is
$f(s)=1$. If all $A{-}a$ is qualified, then pick a minimal qualified
$A'\subseteq A$ and use Lemma~\ref{lemma:2} with any $a\in A'$.

Thus assume $A$ is unqualified. As $\A$ is connected, there is an
unqualified subset $B$ such that $AB$ is qualified (pick any element of $A$
and let $B$ show that this element is important).
Choose such an unqualified $B$ such that the set $B{-}A$ has minimal
cardinality, and within this constrain $A\cap B$ has maximal cardinality. Then
$AB{-}k$ is unqualified for any $k\in B{-}A$ (as otherwise $B{-}A$ is not
minimal), and $aB$ is qualified for any $a\in A{-}B$ (as otherwise $A\cap
B$ is not maximal). Fix $a\in A{-}B$.
With any $k\in B{-}A$ we have that $AB$ is qualified, $AB{-}k$ is not.
Lemma~\ref{lemma:1} gives $f(AB)-f(AB{-}k)=1$, and by Lemma~\ref{lemma:3},
$f(A')-f(A'{-}k)=1$ for all $k\in A'\subseteq AB$. By induction this gives
both $f(AB)-f(A)=|A{-}B|$ and $f(AB{-}a)-f(A{-}a)=|A{-}B|$. Therefore
$f(A)-f(A{-}a) = f(AB)-f(AB{-}a)$. Now $AB$ is qualified. If $AB{-}a$ is
unqualified, then by Lemma~\ref{lemma:1} this difference is $f(s)=1$. If 
$AB{-}a$ is qualified, then $f(AB)=f(AB{-}a)$ using Lemma~\ref{lemma:2} 
with $A'=aB$.
\end{proof}

The main result of this section is the equivalence of a statement about
matroid representability and Conjecture~\ref{conj:dual-of-ideal}. Recall
that $\M$ is an \emph{entropic matroid} if for some positive $\lambda$ the
polymatroid $\lambda\M$ is entropic.

\begin{theorem}\label{thm:main-ideal}
The following statements are equivalent.

\noindent\hangindent\parindent\hangafter=1
\hbox to\parindent{\hss \upshape a)}
The dual of every ideal access structure is ideal.

\noindent\hangindent\parindent\hangafter=1
\hbox to\parindent{\hss \upshape b)}
The dual of every entropic matroid is entropic.
\end{theorem}
\begin{proof}
Let us first make some simplifying assumptions. In a) the access structure
can be assumed to be connected: simply forget about the unimportant
participants, they will be unimportant in the dual structure. In b) the
matroid can be assumed to be tight and connected. This is so as the matroids
$\M$ and $\M\down$ are entropic at the same time: if $i$ has a non-zero
private info, then $i$ is completely independent of the rest of the matroid.
Furthermore, if $\M$ is not connected, then it is an independent sum of the
connected components, and then $\M^\bot$ is the sum of the duals of the
components.

The reduction from entropic matroids to tight entropic matroids was
discussed briefly at the end of Section~\ref{subsec:polymatroids}. In the
proof of Theorem \ref{thm:main-almost-ideal} we need a similar reduction for
almost entropic matroids which is provided by Mat\'u\v s' theorem, see
\cite{entreg}.

\smallskip
a) $\rightarrow$ b)
As remarked above, we may assume that the entropic matroid $\M$ is tight and
connected. Pick any element of its ground set and name it $s$, the remaining
elements are in $P$. Since $\M$ is connected, it has no loops, thus $f(s)=1$.
Define the access structure $\A\subset 2^P$ by
$$
    \A = \{ A\subseteq P:\, f(sA) = f(A) \}.
$$
Clearly $\M$ realizes this access structure, consequently $\A$ is ideal,
and by Claim~\ref{claim:connected-to2} it is also connected. By Claim
\ref{claim:poly-dual-connected} b), the dual matroid $\M^\bot$ realizes
$\A^\bot$.

As $\A$ is a connected ideal structure, assumption a) says that $\A^\bot$ is
ideal. Let $\M'$ be the scaled entropic polymatroid which realizes $\A^\bot$
with $f'(s)=f'(a)=1$. As $\A^\bot$ is connected by Claim
\ref{claim:access-dual-connected}, conditions of Theorem~\ref{thm:b-k}
hold. Consequently $\M'$ is the unique matroid realizing $\A^\bot$. As
$\M^\bot$ also realizes the same access structure, $\M'$ and $\M^\bot$ are
the same matroids. Now $\M'$ is a scaled version of an entropic polymatroid,
thus $\M'=\M^\bot$ is an entropic matroid, as was required.

\smallskip

b) $\rightarrow$ a)
Let $\A$ be an ideal connected access structure realized by the entropic
polymatroid $\M^*$. As $\A$ is connected and ideal, we have $f^*(i)=f^*(s)>0$ for
all participants $i\in P$. Let $\lambda=1/f^*(s)$ and $\M=\lambda\M^*$.
Then $\M$ also realizes $\A$
and $f(i)=f(s)=1$ for all $i\in P$. By Corollary~\ref{corr:ideal-is-tight}
$\M$ is tight, and  by Theorem~\ref{thm:b-k} $\M$ is a matroid. As $\A$ is
connected, by Claim~\ref{claim:connected-to1} $\M$ is connected.
Consequently $\M$ is a tight, connected, entropic matroid which realizes the
access structure $\A$. By assumption b) $\M^\bot$ is an entropic
matroid, realizes $\A^\bot$ by Claim~\ref{claim:poly-dual-connected} b); finally by
Claim~\ref{claim:tight-dual} a) $\M^\bot$ and $\M$ have the same value on
singletons. Thus $\lambda^\bot\M^\bot$ is
an entropic polymatroid for some positive $\lambda^\bot$, realizes
$\A^\bot$, and has complexity $\sigma(\M^\bot)=\sigma(\M)=1$. 
Therefore $\A^\bot$ is ideal.
\end{proof}

Almost entropic polymatroids form a closed cone, which means that positive
multiples of an aent polymatroid are aent. Consequently the definition of
almost entropic matroids does not require scaling as was the case for entropic
matroids. The matroid $\M$ is \emph{almost entropic} if it is almost entropic as a
polymatroid. Repeating the proof above
word by word while replacing ``ideal'' by ``almost ideal'' and ``entropic''
by ``almost entropic'' everywhere one gets the following theorem.

\begin{theorem}\label{thm:main-almost-ideal}
The following statements are equivalent.

\noindent\hangindent\parindent\hangafter=1
\hbox to\parindent{\hss \upshape a)}
The dual of every almost ideal access structure is almost ideal.

\noindent\hangindent\parindent\hangafter=1
\hbox to\parindent{\hss \upshape b)}
The dual of every almost entropic matroid is almost entropic.
\qed
\end{theorem}

\section{Duals of almost entropic matroids}\label{sec:final}

We have almost all the pieces together to prove the main result:
\begin{theorem}\label{thm:main}
There is an almost ideal access structure whose dual is not almost ideal.
\end{theorem}

By Theorem~\ref{thm:main-almost-ideal} we need to exhibit an almost entropic
matroid whose dual is not almost entropic. The existence of such a matroid
was proved by Tarik Kaced~\cite[Theorem~2]{tarik}, this section is a
detailed account of that result. The proof starts with the construction of
an entropic polymatroid whose dual is not entropic. Using a continuity
argument and linear scaling, one gets an integer polymatroid with the same
properties. Theorem~\ref{thm:helgason} established a connection between
integer polymatroids and matroids which preserves duality and almost
entropicity. To complete the tour apply this theorem to get the required
matroid. Now let us see the details.

\smallskip

Finding an entropic polymatroid whose dual is not entropic was a
long-standing open problem. The example below is due to Kaced \cite{tarik}.
The polymatroid is specified by a distribution on five binary random
variables. To show that its dual is not entropic, Kaced used a 5-variable
non-Shannon type information inequality, see \cite{MMRV,fmadhe}. Such an
inequality is a closed halfspace in the $(2^{|M|}-1)$-dimensional space which a)
contains all entropic points on its non-negative side (consequently all aent
points as well), and b) cuts into the polymatroid cone $\GM$. Entropy
inequalities are typically written using
abbreviations originating in information theory. For disjoint subsets $A$,
$B$, $C$ we write 
\begin{align*}
    h(A|B) &= h(AB) - h(B), \\[2pt]
    h(A,B) &= h(A)+f(B)-h(AB), \\[2pt]
    h(A,B\,|\,C) &= h(AC)+h(BC)-h(ABC)-h(C),
\end{align*}
corresponding to conditional entropy, mutual information, and conditional
mutual information, respectively. In any polymatroid these expressions are
always non-negative. The \MMRV{} inequality written for
the singletons of the five-element set $\{abcde\}$ is
\begin{align}\label{eq:mmrv}
  & \big( h(a,b|c)+h(b,c|a)+h(c,a|b)\big) +{} \\
  & + \big(h(b,c|d)+h(b,c|e)+h(d,e)-h(b,c)\big) \ge 0.\nonumber
\end{align}
For a short proof that this inequality holds for aent polymatroids see the
Appendix. 

\begin{claim}\label{claim:tarik}
There is a tight, integer and  aent polymatroid whose dual is not aent.
\end{claim}

\begin{proof}
The distribution on five random variables $\xi_a,\dots,\xi_e$ is
specified in Table~\ref{table:1}. Each of the variables takes either zero or
one, there are only eight combinations with positive probability. The
associated polymatroid $\M_\xi$ is entropic, the left hand side of
(\ref{eq:mmrv}) evaluates to 0.108494. The dual $\M^\bot_\xi$ is not aent as
the left hand side of (\ref{eq:mmrv}) is -0.0715364.

\begin{table}[htb]%
\begin{center}\begin{tabular}{c|c|c|c|c|c}
 $\xi_a$ & $\xi_b$ & $\xi_c$ & $\xi_d$ & $\xi_e$ & Prob \\
\hline
 \rule{0pt}{12pt}0  &  0  &  0  &  0  &  0  & 0.077 \\
 0  &  0  &  1  &  1  &  0  & 0.182 \\
 0  &  1  &  0  &  0  &  1  & 0.182 \\
 0  &  1  &  1  &  0  &  0  & 0.077 \\
 1  &  0  &  0  &  0  &  0  & 0.105 \\
 1  &  0  &  1  &  0  &  0  & 0.136 \\
 1  &  1  &  0  &  0  &  0  & 0.136 \\
 1  &  1  &  1  &  0  &  0  & 0.105 \\
\end{tabular}\end{center}

\kern -12pt

\caption{Distribution on five variables}\label{table:1}
\end{table}%

The duality operation is continuous, thus duals of the polymatroids in a
small neighborhood of $\M_\xi$ are still violating the \MMRV{} inequality.
The entropic polymatroid $\M_\xi$ is on the boundary of the aent cone (for
example, $d$ and $e$ have no private info), but there is another polymatroid
$\M'$ arbitrarily close to $\M_\xi$ \emph{inside} that cone. By
\cite{fmtwocon} interior points of the aent cone are entropic. Take $\M''$
very close to $\M'$ such that all coordinates rational. Let $n$ be the
smallest common denominator of the fractions in the coordinates. Coordinates
in $n\M''$ are integer. The dual of $n\M''$ violates the \MMRV{} as the dual
of $\M''$ violates it, and the left hand side of (\ref{eq:mmrv}) also
multiplies by $n$. Finally, $n\M''$ is entropic: to realize it take $n$
independent copies of the random variables realizing $\M''$. The
tight part of $n\M''$ is integer and almost entropic, its dual is the same
as the dual of $n\M''$ by Claim~\ref{claim:tight-dual} b), proving the
claim.

\smallskip
Using the distribution $\xi$ above, such an integer polymatroid can be constructed
directly. For a subset $A\subseteq M$ define the the polymatroid $\rr_A$ as
$$
    \rr_A: I\mapsto \begin{cases} 1 & \mbox{ if $A\cap I\not=\emptyset$ },\\
                                  0 & \mbox{ otherwise }.
                    \end{cases}
$$
Clearly, $\lambda\rr_A$ is entropic for every positive $\lambda$.
As $A$ runs over all non-empty subsets of $M$ these polymatroids
are linearly independent and span a full-dimensional subcone
of $\GM^*$ consisting of entropic polymatroids only \cite{fmtwocon}. The
idea is that take a multiple of $\M_\xi$ (which is almost entropic),
and use some linear combination of $\rr_A$'s to round up the coordinates 
to integer values. This idea works. The polymatroid
\begin{align*}
 \M& = 
  50.03\M_\xi +  0.3819594(\rr_{abd}+\rr_{acd}+\rr_{abe}+\rr_{ace})+{} \\
 &0.1741526(\rr_b+\rr_c) + 0.0067674\rr_{bc}+ 0.5112645\rr_{abc}+{} \\
 &0.6235703(\rr_{bd}+\rr_{cd}+\rr_{be}+\rr_{ce})+0.0270848(\rr_{bcd}+\rr_{bce})+{} \\
 & 0.1012390\rr_{a}+0.4887355(\rr_{ab}+\rr_{ac})+0.3314441(\rr_{ad}+\rr_{ae})+ {} \\
 &0.3356126(\rr_{abcd}+\rr_{abce}) +
 0.4877698\rr_{bcde}+0.5648560\rr_{abcde}
\end{align*}
is clearly entropic, and it is integer. This is so as the coefficients in
this formula are the solutions of a system of linear equations yielding exact values.
Table~\ref{table:2} shows the coordinates of $51\M_\xi$ (left column), the integer
entropic polymatroid $\M$ (middle column), and the tightening of $\M$ (right
column). The value of the \MMRV{} inequality for the dual of $\M$ is $-1$,
thus $\M^\bot$ is not almost entropic. As $\M^\bot = (\M\down)^\bot$, the
tight part of $\M$ is a tight, integer, aent polymatroid whose dual is not
aent.
\end{proof}

\begin{table}[htb]%
\begin{center}
\begin{tabular}[t]{|rr|r|l|}
\hline
\rule{0pt}{12pt}$a$    &    49.983219  &   55   &   37\\
$b$    &    50.030000  &   55   &   31\\
$c$    &    50.030000  &   55   &   31\\
$d$    &    34.242173  &   38   &   38\\
$e$    &    34.242173  &   38   &   38\\[4pt]
$ab$   &   100.013219  &  107   &   65\\[1pt]
$ac$   &   100.013219  &  107   &   65\\[1pt]
$ad$   &    74.221223  &   81   &   63\\[1pt]
$ae$   &    74.221223  &   81   &   63\\[1pt]
$bc$   &    97.356052  &  105   &   57\\[1pt]
$bd$   &    73.693026  &   80   &   56\\[1pt]
$be$   &    73.693026  &   80   &   56\\[1pt]
$cd$   &    73.693026  &   80   &   56\\[1pt]
$ce$   &    73.693026  &   80   &   56\\[1pt]
$de$   &    65.536972  &   72   &   72\\[1pt]
\hline
\end{tabular}
\qquad
\begin{tabular}[t]{|rr|r|l|}
\hline
\rule{0pt}{12pt}$abc$  &  146.591925  &  155   &   89\\
$abd$  &   111.389648  &  119   &   77\\
$acd$  &   111.389648  &  119   &   77\\
$abe$  &   111.389648  &  119   &   77\\
$ace$  &   111.389648  &  119   &   77\\
$ade$  &    90.946998  &   99   &   81\\
$bde$  &    97.356052  &  105   &   81\\
$cde$  &    97.356052  &  105   &   81\\
$bcd$  &   113.024608  &  121   &   73\\
$bce$  &   113.024608  &  121   &   73\\[2pt]
$abcd$ &   146.591925  &  155   &   89\\
$abce$ &   146.591925  &  155   &   89\\
$abde$ &   122.766078  &  131   &   89\\
$acde$ &   122.766078  &  131   &   89\\
$bcde$ &   128.693164  &  137   &   89\\
$abcde$&   146.591925  &  155   &   89\\
\hline
\end{tabular}
\end{center}

\kern -12pt

\caption{An integer entropic polymatroid}\label{table:2}
\end{table}

Let the tight integer polymatroid provided by Claim~\ref{claim:tarik}
be $\N$, and consider the matroid $\phi(\N)$
provided by Theorem~\ref{thm:helgason}. As $\N$ is aent, $\phi(\N)$ is
almost entropic; $\N^\bot$ is not aent, thus $\phi(\N^\bot)$ is not
aent. Consequently the matroid $\phi(\N)$ is aent and its dual, 
$\phi(\N^\bot)$, is not aent either -- completing the proof of 
Theorem~\ref{thm:main}.

\smallskip
Using the tight almost-entropic polymatroid of Table~\ref{table:2} the
construction in the proof of Theorem~\ref{thm:main-ideal} gives
an almost-ideal access structure on 174 participants (as the corresponding
aent matroid has $f(a)+f(b)+f(c)+f(d)+f(e)=175$ atoms, one of them is the
secret, others are the participants) whose dual is not almost-ideal. It is
left to the interested reader to describe the qualified subsets for different
choices of the secret.

\medskip

According to Theorem~\ref{thm:main-almost-ideal}, to construct a
counterexample to Conjecture~\ref{conj:dual-of-ideal} we need an entropic
matroid whose dual is not entropic. Entropic matroids (and their multiples)
are always on the boundary of the aent cone; the boundary has an intricate
and complicated structure. There seems to be no other way to show that a
matroid is entropic than giving the probability distribution explicitly. But
it is not clear how to guarantee $\H(\xi_A)/\H(\xi_s)$ to be an integer.
No entropic matroid is known which is not a multiple of a linearly
representable polymatroid. Finding such a matroid would be very interesting.

\section{Conclusion and open problems}\label{sec:conclusion}

An almost ideal secret sharing scheme on 174 participants was constructed
explicitly whose dual structure is not almost ideal. It was done by putting
together several pieces of earlier works. The intricate connection between
ideal secret sharing schemes and matroids was observed by Brickell and
Davenport \cite{brickell-davenport}. This connection is expressed in Theorem
\ref{thm:b-k} extending a result of Blakley and Kabatiansky
\cite{blakley-kabatianski}; the presented proof is an adoption of the one
from \cite{fmquantoids}. Entropic and almost entropic polymatroids have been
studied intensively by Frantisek Mat\'u\v s
\cite{M.back,fmadhe,fmtwocon,fminf,entreg}; his insight of the structure of
the entropic region was indispensable to this paper. And, of course, we were
using the surprising result of Tarik Kaced \cite{tarik} who settled an old
conjecture by constructing an entropic polymatroid whose dual is not
entropic. These results allowed us to solve the duality problem of ideal
secret sharing schemes -- is the dual of an ideal structure is ideal? -- in
a model slightly differing from the standard one, namely secret recovery and
secret independence is required only ``up to a negligible factor''. This
model has been studied earlier under the name of ``probabilistic secret
sharing'', see \cite{prob-secret,tarik-2,compartm}. The original problem remained
unsolved.

\begin{problem}
Is the dual of an ideal structure is ideal?
\end{problem}

\noindent
A substantial obstacle attacking this problem was mentioned at the end of
the previous section: every known entropic matroid is linearly representable.

\begin{problem}
Find an entropic matroid which is not linear.
\end{problem}

\noindent
A promising approach seems to be using subgroup representation for the
matroid \cite{chan-yeung}. It requires substantial knowledge of the subgroup
structure of non-commutative finite groups.

\smallskip

It is an interesting question how restrictive the probabilistic model is.
If an access structure has complexity 1, then it is almost ideal 
(while not necessarily ideal \cite{beimel-livne}); this follows from the fact
that the aent polymatroids form a closed cone \cite{fmtwocon}.
The following problem asks about the converse of this implication.

\begin{problem}
Does every almost ideal access structure have complexity 1?
\end{problem}

\noindent
The question is equivalent to whether $\bar\sigma(\A)=1$ implies $\sigma(\A)=1$.
As mentioned at the end of Section \ref{subsec:complexity}, $\bar\sigma$ and
$\sigma$ are not known to be separated.

\begin{problem}
Is there an access structure $\A$ with $\bar\sigma(\A)$ strictly smaller
than $\sigma(\A)$?
\end{problem}

\noindent
This problem was raised in Section \ref{subsec:complexity}:
can the imperfections allowed by the probabilistic model be patched by
adding some small entropy to the secret and / or shares? This question has
been considered by Kaced in \cite{tarik-2}, but remained unsolved.

\section*{Acknowledgment}
The author would like to thank the encouragement and fruitful discussion
at the Prague Stochastics Conference, 2019. He is particularly indebted to
Oriol Farr\'as and Carles Padr\'o. Special thanks to Gabor 
Tardos for figuring out the subtleties of  Claim~\ref{claim:connected-to2}.

The research reported in this paper was supported by GACR project
number 19-04579S, and partially by the Lend\"ulet program of the HAS.

\section*{Appendix}

\begin{theorem}\label{app:them-1}
A matroid is connected if and only if any two points can be connected by a
circuit.
\end{theorem}
\begin{proof}
Let $\M=(h,M)$ be a matroid. Using matroid terminology, $A$ is
\emph{dependent} if $h(A)<|A|$, and $A$
is a \emph{circuit} if it is a minimal dependent set. Every dependent
set contains a circuit. Points $x$ and $y$ are \emph{connected}, written as
$a\approx b$, if there is a circuit containing both of them. First we prove
that $\approx$ is an equivalence relation: if $x\approx z$ and $z\approx y$
then $x\approx y$. This is done in three steps. In claims a), b), c), $C_1$ 
and $C_2$ are different circuits.

\smallskip

a) \emph{Suppose $z\in C_1\cap C_2$. There is a 
circuit $E\subseteq C_1\cup C_2$ which avoids $z$}
(exchange property of circuits): 

\noindent\emph{Proof.}
As $C_1$, $C_2$ are different circuits, $h(C_i)=|C_i|-1$ and
$h(C_1\cap C_2)=|C_1\cap C_2|$ as $C_1\cap C_2$ is a proper subset of a
circuit.  Using submodularity for $C_1$ and $C_2$,
\begin{align*}
   h(C_1\cup C_2) &\le h(C_1)+h(C_2)-h(C_1\cap C_2) \\
      &= |C_1|+|C_2|-h(C_1\cap C_2) -2  \\
      &= |C_1|+|C_2|-|C_1\cap C_2|-2 \\
      &= |C_1\cup C_2|-2.
\end{align*}
Consequently $h(C_1C_2{-}z) \le |C_1C_2{-}z|-1$, which means that
$C_1C_2{-}z$ is dependent, thus contains a circuit.

b) \emph{Let $x\in C_1{-}C_2$, and $z\in C_1\cap C_2$. There is a circuit in
$C_1\cup C_2$ which contains $x$ and avoids $z$.}

\noindent\emph{Proof.} 
By induction on $|C_1\cup C_2|$. By a) there is a circuit $E\subset C_1\cup
C_2$ which avoids $z$. Then $E\cap (C_2{-}C_1)$ is not empty, as $E$ is
dependent 
while $E\cap C_1$, as a proper subset of $C_1$, is independent. If $x\in E$
then we are done. If $x\notin E$, then pick $z'\in E\cap(C_2{-}C_1)$. By a)
there is circuit $F\subset E\cup C_2$ which avoids $z'$. Use induction
on $C_1$ and $F$.

c) \emph{Let $x\in C_1{-}C_2$, $y\in C_2{-}C_1$, and $C_1\cap C_2\not=
\emptyset$. There is a circuit $E\subseteq C_1\cup C_2$ which contains 
$x$ and $y$.}

\noindent\emph{Proof.}
By induction on $|C_1\cup C_2|$. Let $z\in C_1\cap C_2$. 
By b) there is a circuit $E\subset C_1\cup C_2$ which contains
$x$ and avoids $z$. If $y\in E$, then we
are done. If $y\notin E$, then pick $z'\in E\cap(C_2{-}C_1)$. By b) there is
a circuit $F\subset E\cup C_2$ such that $y\in F$ and $z'\notin F$. 
Use induction on $C_1$ and $F$.

\smallskip

This proves that $\approx$ is an equivalence relation. Any two points of the
matroid
are connected by a circuit if and only if there is only a single equivalence
class for $\approx$. First assume that the matroid is connected, and by contradiction that
$A$ is a proper equivalence class of $\approx$. Consider the partition
$A\cup B$ where $B$ is the complement of $A$. Choose the independent
sets $A'\subseteq A$ and $B'\subseteq B$ such that $h(A)=|A'|$ and $h(B)=|B'|$.
As $A$ and $B$ are not independent, 
$h(A'\cup B')\le h(A\cup B)<h(A)+h(B)=h(A')+h(B')=|A'|+|B'|$, thus
$A'\cup B'$ contains a circuit $E$. But $E$ must intersect both $A'$ and
$B'$ (as $A'$ and $B'$ are independent), contradicting that elements from
$A'\subseteq A$ and from $B'\subseteq B$ are not connected.

Conversely, if the matroid is not connected, say the elements of the partition $A\cup B$
are independent, then no circuit can intersect both $A$ and $B$. Indeed,
first the independence of $A$ and $B$ implies $h(E)=h(E\cap A)+h(E\cap B)$
for all subsets $E\subseteq M$.
Second, assume the circuit $E$ intersects both $A$ and $B$. Then
$E\cap A$ and $E\cap B$ are independent (as proper subsets of $E$), and then
$$
  h(E)=h(E\cap A)+h(E\cap B)=|E\cap A|+|E\cap B|=|E|,
$$
a contradicting that $E$ is dependent.
\end{proof}

\begin{theorem}
If $\xi=\langle\xi_a,\dots,\xi_e\rangle$ is a distribution on five elements,
then the polymatroid $\M_\xi$ satisfies the \MMRV{} inequality
\begin{align}\label{eq:rrmv-2}
  & \big( h(a,b|c)+h(b,c|a)+h(c,a|b)\big) +{} \\
  & + \big(h(b,c|d)+h(b,c|e)+h(d,e)-h(b,c)\big) \ge 0,\nonumber
\end{align}
written as $\MMRV(\M_\xi)\ge 0$.
\end{theorem}
\begin{proof}
Observe first that in \emph{any} polymatroid $\M=(h,M)$ the inequality
$$
       \MMRV(\M)+3h(a,de|bc)\ge 0
$$
always holds. This is so as expanding
$\MMRV(\M)+3h(a,de|bc)$ as a linear combination of rank values, and expanding
the clearly non-negative sum below, the results are the same:
\begin{align*}
 &  h(a,d|b)+h(a,d|c)+h(a,e|b)+h(a,e|c)+{} \\
 & + h(b,c|ad) +h(b,c|ae)+h(a,bc|de) +{} \\
 & +h(d,e|a)+ h(a,e|bcd) + h(a,d|bce).
\end{align*}

The \MMRV{} inequality (\ref{eq:rrmv-2}) has been grouped into two parts.
The first part depends only on ranks of subsets of $abc$, and the second
part depends only on subsets of $bcde$. In other words, the value of the
first (and second) part depends only on the marginal distribution
$\xi_{abc}$ and $\xi_{bcde}$, respectively. $\Sigma$ denotes the collection
of all distributions 
$\eta=\langle\eta_a,\dots,\eta_e\rangle$ where each of these five variables
takes the same values as the corresponding variable does in $\xi$ but with 
arbitrary joint probability. Consider the optimization problem of maximizing
the entropy of $\eta\in\Sigma$ under the constraints that certain marginal
distributions are fixed:
$$
  {\textstyle\max_\eta}\,\{  \H(\eta) : \eta\in\Sigma, \eta_{abc}=\xi_{abc},
           \eta_{bcde}=\xi_{bcde} \}.
$$
As $\H(\eta)$ is a strictly convex function of the probabilities, this is a convex
optimization problem with linear constraints, consequently it has a single unique
optimal solution $\eta^*\in\Sigma$. Considering the distribution with the maximal entropy
is often referred to as the 
\emph{maximum entropy principle}. As the marginals on $abc$ and $bcde$ of
$\xi$ and $\eta^*$ are the same, $\MMRV(\M_\xi) = \MMRV(\M_{\eta^*})$. The
extremal distribution $\eta^*$ has the additional property that
$\eta^*_{a}$ and $\eta^*_{de}$ are independent given $\eta^*_{bc}$. This is
so, as fixing the value of $\eta^*_{bc}$, one can redefine the
distribution while keeping the probabilities on $abc$ and on $bcde$ fixed
such that $a$ and $de$ becomes independent.
This would increase the total entropy, thus $a$ and $de$ must be independent
-- giving the claimed conditional independence. Consequently the polymatroid
$\M_{\eta^*}$ satisfies additionally $h^*(a,de|bc) = 0$, and then
$\MMRV(\M_{\eta^*})\ge 0$, proving the theorem.
\end{proof}


\thebibliography{99}

\bibitem{beimel-survey}
A.~Beimel (2011),
Secret-sharing schemes: a survey,
\newblock in: \emph{IWCC 2011, volume 6639 of
LNCS}, Springer, 2011, pp 11-46

\bibitem{beimel-livne}
A.~Beimel, N.~Livne (2006)
On matroids and non-ideal secret sharing
\newblock In: Halevi S., Rabin T. (eds) \emph{Theory of Cryptography, volume
3876 of LNCS}, Springer, Berlin, Heidelberg pp 482-501

\bibitem{blakley-kabatianski}
G.~Blakley, G.~Kabatianski (1995), 
On general perfect secret sharing schemes,
in: \emph{LNCS 963, Advances in Cryptology, Proceedings of Crypto’95}, Springer
1995, pp. 367–371

\bibitem{brickell-davenport}
E.~F.~Brickell, D.~M.~Davenport (1991)
On the classification of ideal secret sharing schemes,
\emph{J. of Cryptology}, vol 4 (73) pp 123-134

\bibitem{prob-secret}
P.~D'Arco, R.~De Prisco. A.~De Santis,
A.~P{\'e}rez del Pozo, U.~Vaccaro (2018),
Probabilistic Secret Sharing,
in: \emph{43rd International Symposium on Mathematical Foundations
of Computer Science, MFCS 2018}, 
Leibniz International Proceedings in Informatics,
Schloss Dagstuhl--Leibniz-Zentrum fuer Informatik,
Vol 117, pp 64:1--64:16


\bibitem{chan-yeung}
T.~H.~chan, R.~W.~Yeung (2002),
On a relation between information inequalities and group theory,
\newblock\emph{IEEE Tran. Information Theory} {\bf 57} pp 6364-6378

\bibitem{fuji}
S.~Fujishige (1978),
Polymatroidal dependence structure of a set of random variables.
\newblock \emph{Information and Control} {\bf 39} 55--72.

\bibitem{helgason}
T.~Helgason (1974)
\newblock Aspects of the theory of hypermatroids,
In: Berge C., Ray-Chaudhuri D. (eds) \emph{Hypergraph Seminar}, Lecture Notes in
Mathematics, vol 411. Springer, Berlin, Heidelberg

\bibitem{error-correcting-codes}
W.~C.~Huffman and V.~Pless (2003),
\emph{Fundamentals of error correcting codes},
\newblock
Cambridge University Press, 2003

\bibitem{tarik}
T.~Kaced (2018),
Information Inequalities are Not Closed Under Polymatroid Duality,
\newblock
\emph{IEEE Transactions on Information Theory}, 64,  pp 4379--4381

\bibitem{tarik-2}
T.~Kaced (2011),
Almost-perfect secret sharing,
\newblock
\emph{Information Theory Proceedings (ISIT), 2011 IEEE International
Symposium on}, pp 1603-1607

\bibitem{katz-lindell}
J.~Katz, Y.~Lindell (2007),
\emph{Introduction to modern cryptography},
Chapman \& Hall/CRC

\bibitem{lovasz}
 L.~Lov\'{a}sz (1982),
Submodular functions and convexity.
\newblock \emph{Mathematical Programming -- The State of the Art} (A.~Bachem,
M.~Gr\"otchel
    and B.~Korte, eds.), Springer-Verlag, Berlin, 234--257.

\bibitem{MMRV}
K.~Makarichev, Y.~Makarichev, A.~Romashchenko, N.~Vereshchagin (2002),
\newblock A new class of non-Shannon type inequalities for entropies.
\newblock
\emph{Communications in Information and Systems}, vol 2, pp 147--166

\bibitem{M.back}
F.~Mat\'{u}\v{s} (1994),
Probabilistic conditional independence structures and matroid theory:
background.
\newblock \emph{Int.\ Journal of General Systems} {\bf 22} 185--196.

\bibitem{fmadhe}
F.~Mat\'u\v s (2007), 
Adhesivity of polymatroids,
\newblock {\it Discrete Mathematics}
{\bf 307} pp. 2464--2477

\bibitem{fmtwocon}
F.~Matúš (2007),
Two constructions on limits of entropy functions. 
\newblock \emph{IEEE Transactions on Information Theory} 53, pp 320-330. 

\bibitem{fminf}
F.~Matúš (2007),
Infinitely many information inequalities.
\newblock \emph{Proceedings IEEE ISIT 2007}, Nice, France, pp 41--44.

\bibitem{fmquantoids}
F.~Matúš (2012),
Polymatroids and polyquantoids. 
in: \emph{Proceedings of WUPES'2012} (eds. J. Vejnarová and T. Kroupa) Mariánské
Lázn\v e, Prague, 
Czech Republic, pp 126-136. 

\bibitem{entreg}
F.~Matúš, L.~Csirmaz (2016),
Entropy region and convolution,
\newblock \emph{IEEE Trans. Inf. Theory} {\bf 62} 6007--6018

\bibitem{oxley}
J.G.~Oxley (1992)
\emph{Matroid Theory}, Oxford Science Publications. The Calrendon Press,
Oxford University Press, New York 

\bibitem{padro-notes}
C.~Padr\'o (2012),
Lecture notes in secret sharing,
\newblock \emph{Cryptology ePrint archive}, report 2012/674

\bibitem{yeung-course} 
R.~W.~Yeung (2002),
\newblock\emph{A First Course in Information Theory},
Kluwer Academic/Plenum Publishers, New York.

\bibitem{compartm}
Y.~Yu, M~Wang (2011), 
A probabilistic secret sharing scheme for a
compartmented access structure.
\newblock\emph{In: International Conference on Information
and Communications Security}, pp 136-142. Springer, Berlin, Heidelberg,
\end{document}